\documentclass[aps,pra,nofootinbib,twocolumn,floatfix,showpacs,superscriptaddress,longbibliography]{revtex4-1}
\usepackage[utf8]{inputenc}
\usepackage{amsmath,graphicx,epsfig,amsthm,color,bm}
\usepackage{pifont,bbding,amssymb,soul,mathrsfs}
\usepackage{indentfirst}
\usepackage{times,dsfont,units}
\usepackage{appendix, comment}

\usepackage{array}
\usepackage{makecell}
\usepackage{booktabs}
\usepackage{diagbox}
\usepackage{threeparttable}
\newtheorem{definition}{Definition}%[section]
\newtheorem{theorem}{Theorem}%[section]
\newtheorem{corollary}{Corollary}%[theorem]

\usepackage{tikz,pgfplots}

\newcommand{\braket}[1]{\langle #1 \rangle}

\newcommand{\tr}{{\rm tr}}
\newcommand{\etal}{\emph{et al.}
}

\newcommand{\bra}[1]{\langle #1|}
\newcommand{\ket}[1]{|#1\rangle}
\newcommand{\bracket}[2]{\langle #1|#2\rangle}

\usepackage{lmodern}
\usepackage[T1]{fontenc}
\usepackage[hyperindex,breaklinks,colorlinks=true,citecolor=blue,urlcolor=blue]{hyperref}

\usepackage[capitalize]{cleveref}

\pgfplotsset{compat=1.17} 
\begin{document}
%\title{Initial-condition-independent maximal entanglement generation in quantum walk without asymptotic limit}%\date{\today}
\title{%
Maximal {coin-position} entanglement {generation in a quantum walk}\\
for the third step and beyond
regardless of the initial state}
\author{Xiao-Xu Fang}
\affiliation{School of Physics, State Key Laboratory of Crystal Materials, Shandong University, Jinan 250100, China}

\author{Kui An}
\affiliation{School of Physics, State Key Laboratory of Crystal Materials, Shandong University, Jinan 250100, China}

\author{Bai-Tao Zhang}
\affiliation{State Key Laboratory of Crystal Materials, Institute of Novel Semiconductors, Shandong University, Jinan 250100, China}

\author{Barry C. Sanders} 
\affiliation{Institute for Quantum Science and Technology, University of Calgary, Alberta, Canada T2N 1N4}

\author{He Lu}
\email{luhe@sdu.edu.cn}
\affiliation{School of Physics, State Key Laboratory of Crystal Materials, Shandong University, Jinan 250100, China}
\affiliation{Shenzhen Research Institute of Shandong University, Shenzhen 518057, China}

\begin{abstract}
We study maximal {coin-position} entanglement generation via a discrete-time quantum walk, in which the coin operation is randomly selected from one of two coin operators set at each step. We solve maximal entanglement generation as an optimization problem with quantum process fidelity as the cost function. Then we determine {the maximal entanglement 
that can be rigorously generated for any step beyond the second regardless of initial condition with appropriate coin sequences. The simplest coin sequence comprising} Hadamard and identity operations {is equivalent to the generalized elephant quantum walk, which exhibits an increasingly faster spreading in terms of probability distribution.} Experimentally, we demonstrate a ten-step quantum walk {driven by} such coin sequences {with linear optics,} and thereby show the desired high-dimensional bipartite entanglement {as well as the transport behavior of faster spreading.}
\end{abstract}

\maketitle

A quantum walk (QW) is the quantum version of a classical random walk~\cite{Aharonov1993PRA,Kempe2003introductory}.
Due to the principle of superposition in quantum mechanics, 
a QW gives rise to impressive applications in quantum information science, from quantum computing~\cite{Childs2004PRA,Childs2009PRL,Lovett2010PRA,Childs2013Science} to quantum simulation~\cite{Schreiber2012Science2DQW}, and from implementing quantum measurement~\cite{Kurzy2013PhysRevLettGeneralizedMeasuring,Bian2015PhysRevLettPOVM,Zhao2015PhysRevA} 
to exploring topological phases~\cite{Kitagawa2010PhysRevAExploring,2012PhysRevB,Kitagawa2012Observation,xiao2017observation,Wang2018PhysRevA,Wang2019PhysRevLett,Xu2019PhysRevResearch,WZS19}.
For the discrete-time QW~(DTQW), entanglement can be generated between {coin and position} degree of freedom of the walker, {so called the coin-position entanglement}~\cite{Carneiro2005Entanglement,Abal2006PRA,Annabestani2010JPAMT,Franco2011PRL}, which is a key resource for quantum information processing~\cite{Horodecki2009RevModPhys}. {The entangled states generated in DTQW are generally high-dimensional quantum states ($2\otimes d$) that exhibit contents richer than those of qubit states ($2\otimes 2$)~\cite{Erhard2020}. Thus, DTQW provides an experimental platform to investigate quantum correlations of $2\otimes d$ quantum states in terms of separability and entanglement detection~\cite{Ha2005PhysRevA,Zhao2011PhysRevA_Endetection,Chen2012PhysRevA,Johnston2013PhysRevA,Shen2020PhysRevA}}, entanglement of formation~\cite{Gerjuoy2003PhysRevA,Lastra2012PhysRevA}, survival of entanglement~\cite{Dajka2008PhysRevA,Giordani2021NJP},  concurrence~\cite{Mintert2004PhysRevLett,Chen2005PhysRevLett,Zhao2011PhysRevA} and discord~\cite{Sai2012,Girolami2012PhysRevLett,Ma2015SR}. 

In a {one-dimensional} (1D) DTQW with static coin operations (unchanging coin operation during evolution), 
entanglement generation depends on the initial coin state and cannot reach the maximal value~\cite{Carneiro2005Entanglement,Abal2006PRA}.
Counterintuitively, by introducing disorder into the DTQW~\cite{Chandrashekar2013ArXiv}, e.g., randomly choosing SU(2) coin operation
\begin{equation}\label{Eq:SU2coin}
\hat{C}(\xi,\gamma,\zeta)
=\begin{pmatrix}\text{e}^{\text{i}\xi}\cos\gamma&\text{e}^{\text{i}\zeta}\sin\gamma\\\text{e}^{-i\zeta}\sin\gamma&-\text{e}^{-i\xi}\cos\gamma\end{pmatrix},
4\gamma,\xi,\zeta\in[0, 2\pi],
\end{equation}
at each step,
generated entanglement is significantly enhanced and achieves maximal entanglement generation~(MEG) asymptotically
independent of initial conditions~\cite{Vieira2013PhysRevLett,Vieira2014PhysRevA}. 
Motivated by robust entanglement generation under experimental conditions with imperfections and disorder,
random-coin DTQWs have been theoretically studied for various disorder configurations~\cite{Salimi2012,Rohde2013PRA,Montero2016PRA,Molfetta2016,Zeng2017,Orthey2019,Singh2019,Buarque2019PhysRevE,Pires2021,Laneve2021PhysRevResearch}
and have been experimentally observed with linear optics~\cite{Wang2018Optica,Tao2021OL,Zhang2022PRA}. {The enhancement of entanglement is not limited to the disorder in coin operation. Introducing the disorder in the shift operator can enhance the coin-position entanglement generation as well~\cite{Pires2019SR,SEN2020,Pires2020PRE_shiftdisorder,Naves2022PRA}. Besides, phenomena of entanglement boosting also exist in the quantum Parrondo walk~\cite{Pires2020PRE,Munsif2020AQT,Walczak2021PRE,Walczak2022PRE_threecoin} }

{The MEG with a fixed initial coin state can be obtained either in the asymptotic approach~\cite{Pires2019SR,SEN2020} or with specifically designed coin-operation sequences~(henceforth called a coin sequence)~\cite{Gratsea_2020_Generation_of_hybrid,Tao2021OL,Zhang2022PRA}.
MEG, regardless of initial coin states, is generally achieved in an asymptotic approach via QW with disorder either in coin operations~\cite{Vieira2013PhysRevLett} or in shift operations~\cite{Naves2022PRA}}, which is problematic for current experimental technologies.
Strategies to optimize coin sequences during the evolution have been proposed aiming at MEG for few steps. Universal coin sequences are proposed to generated highly entangled states for fewer than ten steps~\cite{Gratsea_2020_Universal_optimal_coinSequences}. However, the universal sequence works for an odd number of steps and for the states with vanishing relative phase. Parrondo sequences have been proposed to generate maximal entanglement at steps $T=3$ and $T=5$~\cite{Dinesh2022}. 
Ideal MEG via a DTQW should work for any step number \emph{and} independent of initial conditions,
but previous experiments have achieved either one or the other, not both;
we achieve both simultaneously here
for all steps beyond the second
by solving an optimization problem. {Interestingly, the determined optimal coin sequences are equivalent to the generalized elephant quantum walk~(gEQW)~\cite{Pires2019SR,Naves2022PRA}, in which the spreading of the probability distribution is much faster.}
Experimentally, we demonstrate the DTQW with requisite coin sequences up to ten steps with linear optics
and observe significant enhancement of entanglement generation {as well as spreading behavior} compared to other schemes.
\begin{figure*}[ht!]%4.25
 	\includegraphics[width=\linewidth]{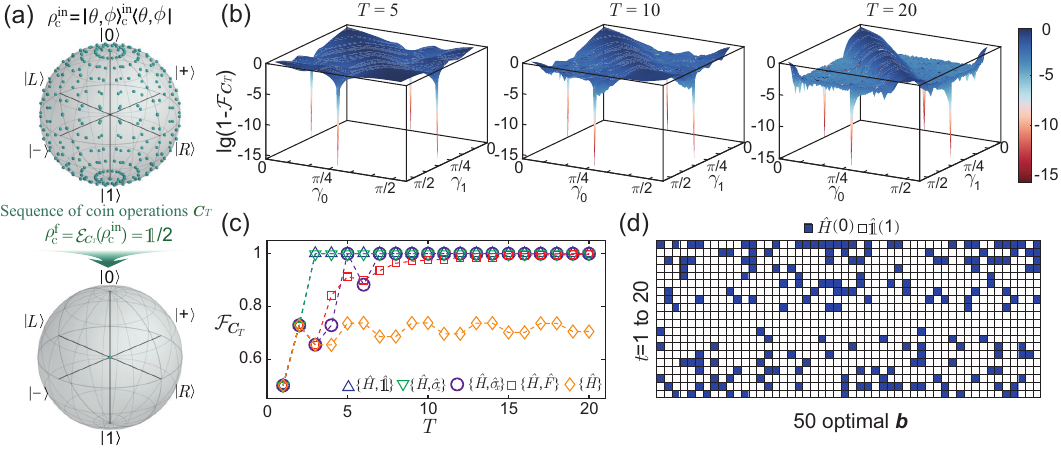} 
 	\caption{(a) The geometric representation of the optimal coin sequence $\bm{C}_T$ that can generate maximal entanglement irrelevant of initial state $\ket{\theta,\phi}_\text{c}^{\text{in}}$. (b) Results of {optimization of $1-\mathcal F_{\bm{C}_T}$ with $\hat{C}_t\in\{\hat{C}(\gamma_0), \hat{C}(\gamma_1)\}$} at $T=5$, $T=10$ and $T=20$. {The values of $\gamma_(0,1)$ are taken from 0$^\circ$ to 90$^\circ$ with an interval of 1$^\circ$.} (c) The maximal $\mathcal F_{\bm{C}_T}$ in a QW from $T=1$ to $T=20$ {with coin set $\{\hat{H}, \hat{\mathds1}\}$ (blue up‑pointing triangle), $\{\hat{H}, \hat{\sigma}_z\}$ (green down‑pointing triangle), $\{\hat{H}, \hat{\sigma}_x\}$ (purple circle), $\{\hat{H}, \hat{F}\}$ (red square) and $\{\hat{H}\}$ (yellow diamond)}.  (d) Fifty $\bm b$s {that can achieve $\mathcal F_{\bm{C}_T}=1$} at $T=20$ with coin set $\{\hat{H}, \hat{\mathds1}\}$ where 0 represents $\hat{H}$ and 1 represents $\hat{\mathds1}\}$.} 
 	\label{Fig:conceptandoptimazation}
 \end{figure*}
 
In the 1D DTQW, the Hilbert space of coin (c) and position (p) of the walker is $\mathscr{H}
    =\mathscr{H}_\text{c}\otimes\mathscr{H}_{\text{p}}$
with
\begin{equation}
\label{eq:Hcw}
    \mathscr{H}_\text{c}=\text{span}\{\ket0_\text{c},\ket1_\text{c}\},\,
    \mathscr{H}_\text{p}=\text{span}\{\ket{x};
x \in \mathbb{Z}\}.
\end{equation}

The walker is initially localized {in position state $\ket0_\text{p}$},
with arbitrary initial coin 
$\ket{\theta,\phi}_\text{c}^{\text{in}}=\cos(\nicefrac\theta2)\ket0_\text{c}+\text{e}^{\text{i}\phi}\sin(\nicefrac\theta2)\ket1_\text{c}$,
where $2\theta,\phi \in [0,2\pi]$. At step~$t$, the coin operator $\hat{C}_t$ is applied.
Then the walker moves left or right {conditioned on the coin state by}
\begin{equation}
\label{eq:shiftoperator}
\hat{S}=\sum_{x}\ket{x+1}_\text{p}\bra{x}\otimes\ket0_\text{c}\bra{0}+\ket{x-1}_\text{p}\bra{x}\otimes\ket1_\text{c}\bra{1},
\end{equation}
which is independent of~$t$.

\begin{figure*}[htb]%4.25
 	\includegraphics[width=\linewidth]{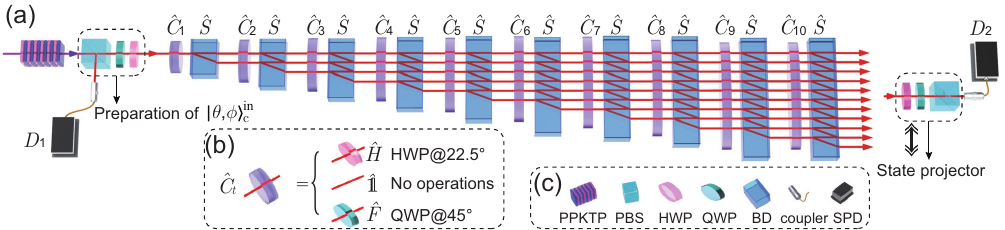} 
 	\caption{(a) Detailed sketch of the setup to realize the ten-step DTQW. (b) Coin operations realized in experiment. A half-wave plate~(HWP) set at 22.5$^\circ$ corresponds to operation $\hat{H}$ and a quarter-wave plate~(QWP) set at 45$^\circ$ corresponds to operation $\hat{F}$. No waveplate needs to be arranged if the operation is $\hat{\mathds1}$. (c) Symbols used in panels(a) and (b): periodically poled potassium titanyl phosphate (PPKTP), polarization beam splitter (PBS), half-wave plate (HWP), quarter-wave plate (QWP), beam displacer (BD), and single photon detector (SPD).} 
 	\label{Fig:Setup}
 \end{figure*} 

For $t\in[N]=\{1,\ldots,N\}$,
the evolution is
\begin{equation}
\label{eq:evolution}
\ket{\theta,\phi}_{\text{f}}
=\prod_{t\in[T]}\hat{U}_t\ket{\theta,\phi}_\text{c}^{\text{in}}\otimes\ket0_{\text{p}},\,
\hat{U}_t=\hat{S}(\hat{C}_t\otimes\hat{\mathds1}_{\text{p}}),
\end{equation}
where $\hat{\mathds1}_{\text{p}}=\sum_{x}\ket{x}\bra{x}$
is the identity operator on $\mathscr{H}_{\text{p}}$,
and~"f" is short for "final".
The sequence
$\bm{C}_T=\left(\hat{C}_t\right)_{t\in[T]}$
describes coin operations applied to the walker. 

Achieving {coin-position} MEG at step~$T$ regardless of $\ket{\theta,\phi}_\text{c}^{\text{in}}$ corresponds to designing $\bm{C}_T$ that maps any $\ket{\theta,\phi}_\text{c}^{\text{in}}\otimes\ket0_{\text{p}}$ to the maximally entangled {coin-position} state $\ket{\theta,\phi}_{\text{f}}$.
Entanglement of $\ket{\theta,\phi}_{\text{f}}$ 
is quantified by the {Von Neumann} entropy
\begin{equation}
\mathcal S_{\text{E}}(\ket{\theta,\phi}_\text{f})
=-\tr\left(\rho_\text{c}^{\text{f}} \log_2 \rho_\text{c}^{\text{f}}\right)
=-\sum_{\varepsilon\in\pm}
\lambda_
\varepsilon\log_2\lambda_\varepsilon
\label{Eq:vonneumannentropy}
\end{equation} 
of the reduced coin state~\cite{Abal2006PRA,Bennett1996PhysRevA}
$\rho_\text{c}^{\text{f}}=\tr_{\text{p}}(\ket{\theta,\phi}_\text{f}\bra{\theta,\phi})$ and $\lambda_{\pm}$ are the eigenvalues of $\rho_\text{c}^{\text{f}}$.
Note that $0\leq \mathcal S_{\text{E}}\leq 1$, and $\mathcal S_{\text{E}}\equiv0$ for separable states and~1 for maximally entangled states.

Thus, MEG evolution~(\ref{eq:evolution}) yields maximally entangled $\ket{\theta,\phi}_{\text{f}}$,
which is equivalent to  $\mathcal E_{\bm{C}_T}(\rho_\text{c}^{\text{in}}=\ket{\theta,\phi}_\text{c}^{\text{in}}\bra{\theta,\phi})=\nicefrac{\mathds1}{2}$ in~$\mathscr{H}_\text{c}$, where $\mathcal E_{\bm{C}_T}$ is a completely-positive linear map determined by $\bm{C}_T$.
A geometric illustration of~$\mathcal E_{\bm{C}_T}$ for MEG
is in~\cref{Fig:conceptandoptimazation}(a), which is the depolarizing channel $\mathcal E^{\text{DP}}(\rho_\text{c}^{\text{in}})=(1-\eta)\rho_\text{c}^{\text{in}}+\eta\nicefrac{\mathds1}{2}$ with $\eta=1$~\cite{Nielsen2011Quantum}.
Process fidelity
{$\mathcal F_{\bm{C}_T}
=\tr\left(\sqrt{\sqrt{\chi_{\bm{C}_T}}\chi_{\text{DP}}\sqrt{\chi_{\bm{C}_T}}}\right)^2$~\cite{Bongioanni2010PRA,Wilde2017}},
is our figure of merit to design $\bm{C}_T$,
where $\chi_{\bm{C}_T}$ is the Pauli-matrix representation of the quantum channel $\mathcal E_{\bm{C}_T}$.
Note that $\mathcal F_{\bm{C}_T}=1$ indicates MEG at step~$T$ regardless of $\ket{\theta,\phi}_\text{c}^{\text{in}}$, and we refer to the corresponding coin sequence $\bm{C}_T$ as the \emph{optimal} coin sequence. {In this sense, the design of optimal $\bm{C}_T$ can be addressed by solving the optimization problem
\begin{equation}
\label{Eq:opt}
\begin{split}
\text{maximize}\hspace{1.2cm}&\mathcal F_{\bm{C}_T}
=\tr\left(\sqrt{\sqrt{\chi_{\bm{C}_T}}\chi_{\text{DP}}\sqrt{\chi_{\bm{C}_T}}}\right)^2\\
\text{subject to}\hspace{1.2cm}&\hat{C}_t\in \text{SU}(2).
    \end{split}
\end{equation}
A general SU(2) coin operation in~\cref{Eq:SU2coin} has three parameters, which makes the optimization rather resource demanding. To simplify the optimization, we }replace
$\hat{C}(\gamma)\gets\hat{C}(0,\gamma,0)$.
Furthermore, we restrict construction of~$\bm{C}_T$
by allowing only two coin operations, i.e., $\gamma_{0,1}$
labeled by one bit with values~0 and~1.
{Then the optimization~\cref{Eq:opt} converts to
\begin{equation}\label{Eq:opt_simplified}
\begin{split}
\text{maximize}\hspace{1.2cm}&\mathcal F_{\bm{C}_T}
=\tr\left(\sqrt{\sqrt{\chi_{\bm{C}_T}}\chi_{\text{DP}}\sqrt{\chi_{\bm{C}_T}}}\right)^2\\
\text{subject to}\hspace{1.2cm}&\hat{C}_t\in \{\hat{C}(\gamma_0), \hat{C}(\gamma_1)\}, \gamma_{0,1}\in [0, \nicefrac{\pi}{2}].
    \end{split}
\end{equation}
  }

We solve optimization~(\cref{Eq:opt_simplified}) using an annealing algorithm.
The results of optimization of $\gamma_{0,1}$
at $T\in\{5,10,20\}$
are shown in~\cref{Fig:conceptandoptimazation}(b).
Evidently, the minimal $1-\mathcal F_{\bm{C}_T}$
is obtained for two coin sets:
$\{\hat{C}(0), \hat{C}(\nicefrac{\pi}4)\}$ and $\{\hat{C}(\nicefrac{\pi}2), \hat{C}(\nicefrac{\pi}4)\}$. Note that $\hat{C}(0)=\hat{\sigma}_z$, $\hat{C}(\nicefrac{\pi}4)=\hat{H}$ and $\hat{C}(\nicefrac{\pi}2)=\hat{\sigma}_x$.
The evolution unitary operator with the coin operator $\hat{\sigma}_z$, i.e.,  $\hat{U}=\hat{S}\hat{\sigma}_z$,
makes the components~$\ket0_\text{c}$ and~$\ket1_\text{c}$ propagate in the opposite direction without interference, which has the similar effect of $\hat{U}=\hat{S}\hat{\mathds1}$. {The difference is that $\hat{\sigma}_z$ delivers a phase $\pi$ ($\text{e}^{\text{i}\pi}=-1$) on component~$\ket1_\text{c}$ while $\mathds 1$ delivers zero phase ($\text{e}^{\text{i}0}=1$), which does not affect the amount of entanglement of the final state. Along this spirit, we conjecture that the coin set $\{\hat{H}, \hat{\mathds1}\}$ is as effective as $\{\hat{H}, \hat{\sigma_z}\}$ in terms of MEG. To confirm this conjecture, we solve~\cref{Eq:opt_simplified} by restricting $\hat{C}_t\in\{\hat{H}, \hat{\mathds1}\}$ and $\hat{C}_t\in\{\hat{H}, \hat{\sigma_z}\}$ respectively, and the results of optimized $\mathcal F_{\bm{C}_T}$ with~$T$ up to~20 are shown with blue up-pointing triangles and green down-pointing triangles in~\cref{Fig:conceptandoptimazation}(c).

We observe that optimized $\mathcal F_{\bm{C}_T}$ values with these two coin sets are exactly same, in which $\mathcal F_{\bm{C}_T}=1$ since step $T=3$.} To give a comparison, we also show the optimized $\mathcal F_{\bm{C}_T}$ with coin sets $\{\hat{H}, \hat{\sigma}_x\}$ and $\{\hat{H}\}$. {We also consider the coin set of $\{\hat{H}, \hat{F}\}$ with} $\hat{F}=[1,\text{i};\text{i}, 1]/\sqrt{2}$ {being the Kempe coin operator~\cite{Kempe2003introductory}, which is widely adopted in the investigation of entanglement generation in discorded QW~\cite{Vieira2013PhysRevLett,Wang2018Optica,Orthey2019,Gratsea_2020_Universal_optimal_coinSequences}. As shown in~\cref{Fig:conceptandoptimazation}(c),} $\mathcal F_{\bm{C}_T}=1$ is achieved at step $T=5$ and $T\geq 7$ for the coin set $\{\hat{H}, \hat{\sigma}_x\}$ {(purple circles)}.
Asymptomatic behavior is observed with the coin set $\{\hat{H}, \hat{F}\}$ {(red squares)} and oscillating behavior is observed in the Hadamard walk~{(yellow diamonds)}. {The optimized $\mathcal F_{\bm{C}_T}$ is associated with a bit string $\bm b\in\{0,1\}^T$ of length~$T$ with~$0$ labeling~$C(\gamma_0)$ and~$1$ labeling~$C(\gamma_1)$.} We note that optimal~$\bm b$ at step~$T$ is not unique. For instance, we obtain 1104 optimal~$\bm b$ with the coin set $\{\hat{H}, \hat{\mathds1}\}$ at $T=20$, and we list 50 among them in~\cref{Fig:conceptandoptimazation}(d).
There are no obvious features of regularities and generalities of these optimal~$\bm b's$.
An optimal~$\bm b$ containing~0 ($\hat{H}$) as little as possible is preferred in experiment. {Considering the spreading behavior with the optimal coin sequences, we experimentally choose the optimal~$\bm b$ containing two or three~0s in our realization} (The explicit form of $\bm b$ and its corresponding proof are given in~\cref{App:optimalsequence}). {Note that the optimal $\bm b$ generally guarantees the MEG at step $T$. However, there indeed exists optimal $\bm b$, such as the sequences in~\cref{App:optimalsequence}, which leads to MEG at intermediate steps as well.}

\begin{figure*}[ht!bp]%4.25
 \includegraphics[width=\linewidth]{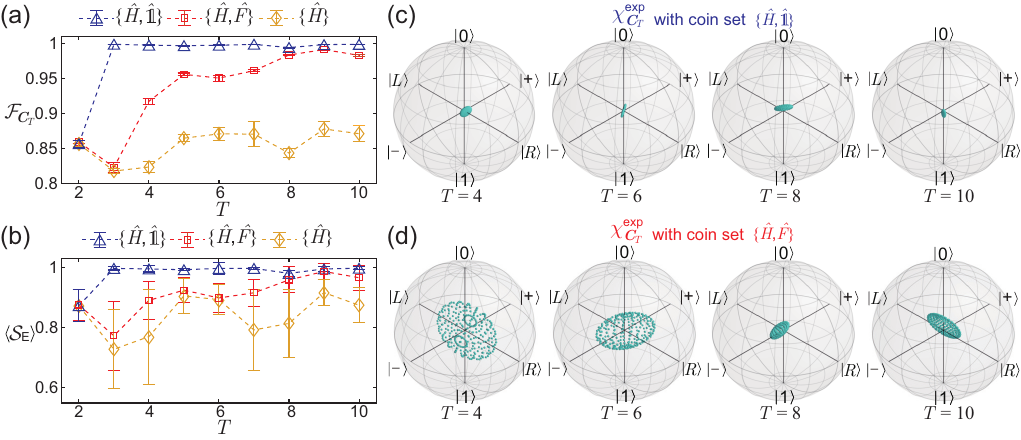}
\caption{%
(a) Experimental results of $\mathcal F_{\bm{C}_T}$
 with coin sets $\{\hat{H}, \hat{\mathds1}\}$ {(blue up-pointing triangles)}, $\{\hat{H},\hat{F}\}$ {(red squares)} and $\{\hat{H}\}$ {(yellow diamonds)} at $T=2$ to $T=10$.
(b) Average entanglement $\langle\mathcal{S}_\text{E}\rangle$ over 296 initial coin states with the reconstructed
$\chi^\text{exp}_{\bm{C}_T}$.
(c) Geometric representation of reconstructed $\chi^{\text{exp}}_{\bm{C}_T}$ with the coin set $\{\hat{H}, \hat{\mathds1}\}$ at $T=4,6,8$ and $10$.
(d) Geometric representation of reconstructed $\chi^{\text{exp}}_{\bm{C}_T}$
with the coin set
$\{\hat{H},\hat{F}\}$
at $T=4,6,8$ and $10$.}
\label{Fig:Geometricplot}
\end{figure*}

{In fact, the disorder QW with the coin set $\{\hat{H}, \hat{\mathds1}\}$ is equivalent to the gEQW~\cite{Pires2019SR,Naves2022PRA}, which is a QW with disorder in the shift operator. In gEQW, the coin operator is step-independent (static coin), and the shift operator is step-dependent according to
\begin{equation}\label{Eq:shiftoperator_gEQW}
\begin{split}
&\hat{S}_{\text{gEQW}}(\Delta_t)\\
&=\sum_{x}\ket{x+\Delta_t}_\text{p}\bra{x}\otimes\ket0_\text{c}\bra{0}+\ket{x-\Delta_t}_\text{p}\bra{x}\otimes\ket1_\text{c}\bra{1}
\end{split}
\end{equation} 
with $\Delta_t\in[1, 2, \ldots, T]$. The probability distribution of $\Delta_t$ is a discretized version of the $q$-exponential distribution~\cite{Tsallis2009nonadditive}. As a consequence, the case of $q=0.5$ corresponds to a standard Hadamard QW while for $q\to\infty$ the shift operation~\cref{Eq:shiftoperator_gEQW} becomes completely disordered.

In our mode, the evolution of the $l$-step QW with a coin sequence comprised of a single Hadamard operation followed by $l-1$ identity operations corresponds to a one step gEQW with $\Delta_t=l$
\begin{equation}
\prod_{\bm b=01\cdots1}\hat{S}(\hat{C}_t\otimes\mathds1_\text{p})=\hat{S}_{\text{gEQW}}(\Delta_1=l)\hat{H}\otimes\mathds1_\text{p}.
\end{equation}
For instances, the optimal $\bm b$ in the first column of~\cref{Fig:conceptandoptimazation}(d) is $\bm b=01111111111010111011$, which corresponds to a four-step gEQW with $\Delta_1=11$, $\Delta_2=2$, $\Delta_3=4$ and $\Delta_4=3$ (See~\cref{App:gEQW} for more details). Compared to the QW with disorder in the coin operation, the gEQW exhibits a faster spreading while maintaining the capability of asymptotic MEG~\cite{Pires2019SR,Naves2022PRA}.   }

{We implement the 1D DTQW with the well-established dynamical evolution of single photon in linear optical network~\cite{Broome2010PRL,Xue2014NJP,Xue2015PRL}.} The experimental setup is shown in~\cref{Fig:Setup}(a). The coin state is encoded in the photon's polarization degree of freedom by $\ket{H(V)}=\ket{0(1)}$, where $\ket{H(V)}$ denotes the horizontal (vertical) polarization.
The position state is encoded in the photon's spatial degree of freedom, i.e., the transverse spatial modes.
Two photons in state $\ket{H}\ket{V}$ with a central wavelength at 810~nm are generated from a periodically poled potassium titanyl phosphate (PPKTP) crystal pumped by an ultraviolet {continuous-wave} laser diode with the central wavelength at 405~nm~\cite{Li2021PRR,Ding2021PRR,Zhang2021PRL}.

{During our experiment, the count rate of two-photon coincidences is about $2.8\times10^4$/s with a pump power of 10mW.} The two photons are then separated by a polarizing beam splitter (PBS), which transmits the horizontal polarization and reflects vertical polarization.
The reflected photon is detected by a single-photon detector (SPD) to serve as a trigger. The transmitted photon is sent into the photonic network consisting of waveplates and birefringent calcite beam displacers (BDs), in which the longitudinal spatial mode of the injected photon is denoted as the start position of the walker $\ket0_\text{p}$. The coin operations $\hat{C}_t$ are realized by waveplates which rotates the polarization of the photon, and the BD transmits the vertical polarization while deviating from the horizontal polarization so that the BD acts as the shift operation $\hat{S}$.

{By carefully adjusting the position between any pair of two BDs, we observe an average interference visibility beyond 0.99. Note that if no waveplate is set between two BDs, the concatenation of two BDs corresponds to the shift operator~(\ref{Eq:shiftoperator_gEQW}) with $\Delta_t=2$. The outgoing state is detected by a state analyzer as shown in~\cref{Fig:Setup}(a). The projective measurement on the position state $\ket{x}$ is achieved by placing a SPD at the corresponding output mode of the optical network, and the projective measurement on an arbitrary coin state is implemented by a HWP, a QWP and a PBS. } 

\begin{figure*}[htbp]%4.25
\includegraphics[width=\linewidth]{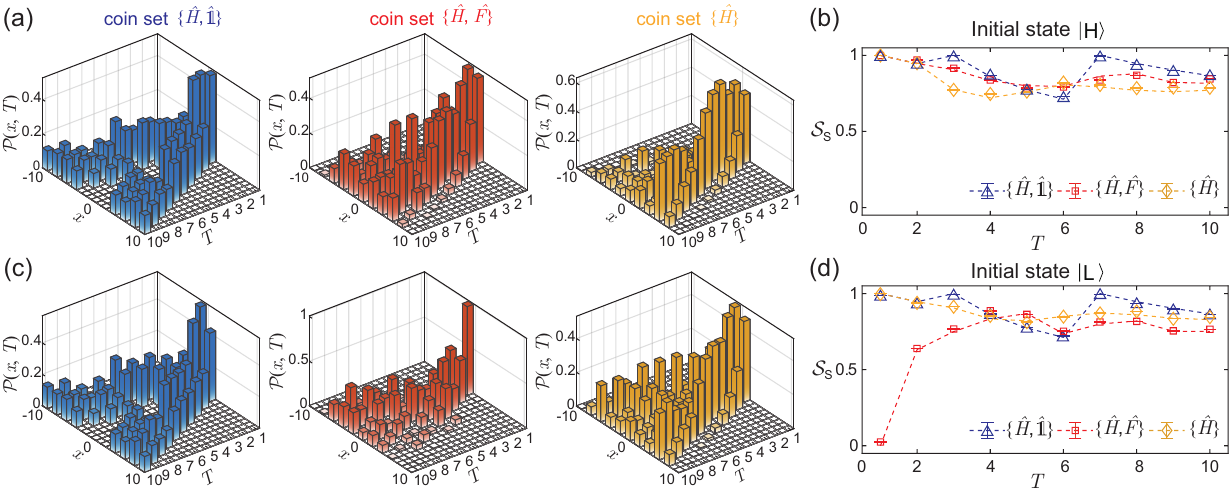} 
\caption{%
Measured $\mathcal P(x,T)$ with initial coin states (a) $\ket{H}$ and (c) $\ket{L}$ {driven by the coin sequences with coin sets $\{\hat{H}, \hat{\mathds1}\}$, $\{\hat{H}, \hat{F}\}$ and $\{\hat{H}\}$, respectively.} The calculated $\mathcal S_\text{S}(T)$ with measured $\mathcal P(x,T)$ of input state (b) $\ket{H}$ and (d) $\ket{L}$, {where the results of coin sets $\{\hat{H}, \hat{\mathds1}\}$, $\{\hat{H}, \hat{F}\}$ and $\{\hat{H}\}$ are shown with blue up-pointing triangles, red squares and yellow diamonds respectively.}} 
\label{Fig:spread}
\end{figure*}

To reconstruct the process matrix $\chi^{\text{exp}}_{\bm{C}_T}$,
we prepare four states as initial coin states $\ket{\theta,\phi}_\text{c}^{\text{in}}$, i.e., $\ket{H}$, $\ket{V}$, $\ket{+}\left(1/\sqrt{2}\left(\ket{0}+\ket{1}\right)\right)$ and $\ket{L}\left(1/\sqrt{2}\left(\ket{0}+i\ket{1}\right)\right)$. For each step~$T$, we set the the optimal coin sequence $\bm{C}_T$ accordingly (see~\cref{App:optimalsequence} for the settings of coin sequences), and we reconstruct $\rho_\text{c}^{\text{f}}$ using quantum state tomographic technology~\cite{Nielsen2011Quantum}. {To this end, we first set the measurement apparatus at one mode of the output of the optical network, and we perform the projective measurement on coin states $\ket{H}, \ket{V}, \ket{+}$ and $\ket{L}$ respectively. Then we move the measurement apparatus to the next optical mode and repeat the process of projective measurements aforementioned.}

{After collecting the data over all optical modes, we put the data together to perform quantum state tomography without distinguishing which mode they come from, which corresponds to trace out of position DOF. Roughly $2.2\times10^5$ two-photon coincidences are collected to perform process tomography at each step.} The experimental results of $\mathcal F_{\bm{C}_T}$ with the coin set $\{\hat{H}, \hat{\mathds1}\}$ are shown with blue triangles in~\cref{Fig:Geometricplot}(a).
We observe that the average $\mathcal F_{\bm{C}_T}$ from $T=3$ to $T=10$ is $0.9954\pm0.0008$, which is much better than the results with the coin set $\{\hat{H}, \hat{F}\}$ as shown with red squares. 
For the Hadamard QW, $\mathcal F_{\bm{C}_T}<0.8$
and oscillates as~$T$ increases (shown with yellow diamonds). We calculate the average entanglement $\langle\mathcal{S}_\text{E}\rangle$ over 296 initial coin states with the reconstructed $\chi^\text{exp}_{\bm{C}_T}$, and the results are shown in~\cref{Fig:Geometricplot}(b). The error bar indicates initial-state-independence, and we observe a stronger initial-state-independence with the coin set $\{\hat{H}, \hat{\mathds1}\}$ with the other two. This is also reflected by the geometric interpretations of $\mathcal F_{\bm{C}_T}$ as shown in~\cref{Fig:Geometricplot}(c) (coin set $\{\hat{H}, \hat{\mathds1}\}$) and~\cref{Fig:Geometricplot}(d) (coin set $\{\hat{H}, \hat{F}\}$) at $T=4, 6, 8$ and $10$, respectively. It is obviously that the results with $\{\hat{H}, \hat{\mathds1}\}$ are much more dense than the results with $\{\hat{H}, \hat{F}\}$, which indicates the entanglement generation with the coin set $\{\hat{H}, \hat{\mathds1}\}$ has stronger independence of the initial coin states stronger than that with the other two. More details are shown in~\cref{App:processodd}. 

\begin{figure*}[htbp!]
\includegraphics[width=1\linewidth]{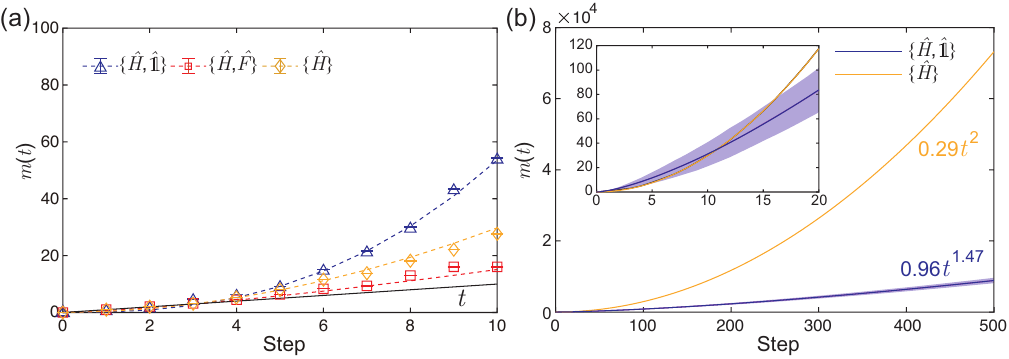} 
\caption{(a) The experimental results of the second moment $m(t)$ in a ten-step QW. {The blue up-pointing triangles represent the results with the coin set $\{\hat{H},\hat{\mathds1}\}$. The red squares and yellow diamonds represent the results with the coin set $\{\hat{H},\hat{F}\}$ and $\{\hat{H}\}$, respectively}. The black line is   $m(t)=t$. (b) The simulated 500-step QW with the coin sets $\{\hat{H},\hat{\mathds1}\}$ and $\{\hat{H}\}$. For QW with the coin set $\{\hat{H},\hat{\mathds1}\}$, the coin operation is randomly selected from the coin set $\{\hat{H},\hat{\mathds1}\}$ at etch step. The average of $m(t)=t$ (blue line) is calculated on a sample of 1000 different $\bm{C}_T$ values, and the blue shade corresponds to the standard deviation of $m(t)=t$. The insert is $m(t)=t$ from $t=1$ to $t=20$.}
\label{Fig:smoent_exp}
\end{figure*}

{We investigate the spreading properties of  the demonstrated QW. We first }investigate the uniformity of the probability distribution $\mathcal P(x, T)$ {at step~$T$, which} can be characterized by the normalized Shannon entropy 
\begin{equation}
\mathcal S_\text{S}(T)=\frac{-\sum_x \mathcal P(x, T) \ln\mathcal P(x, T)}{\ln(T+1)}, 
\end{equation}  
with $\nicefrac1{\ln(T+1)}$ being the normalization parameter.
The walker is able to occupy $T+1$ positions after~$t$ steps so that the maximal value of $-\sum_x \mathcal P(x,T) \ln\mathcal P(x, T)$ is $\ln (T+1)$, which corresponds to the uniform distribution over $T+1$ positions~\cite{2020_martin-vazquez_Phys.Rev.A_Optimizing}.
Larger $\mathcal S_\text{S}(T)$ implies $\mathcal P(x, T)$ is more uniform. For a $T$-step QW associated with the corresponding optimal $\bm{C}_T$, we measure the probability distribution~$\mathcal P(x,T)$ at step~$T$, according to which we calculate the normalized Shannon's entropy $\mathcal S_\text{S}(T)$. The results of $\mathcal P(x,T)$ with the initial coin states $\ket{H}$ and $\ket{L}$ are shown in~\cref{Fig:spread}(a,c), and the corresponding $\mathcal S_\text{S}(T)$ values are shown in~\cref{Fig:spread}(b,d) respectively. Compared with the other two cases, the uniformity of the QW with the coin set $\{\hat{H}, \hat{\mathds1}\}$ is enhanced at $T=3$ and $T=7$.  We also investigate the trend of probability distributions, which can be indicated by the second moment of the walker
\begin{equation}
m(t)=\sum_x x^2\mathcal P(x, t).
\end{equation} 
The walker shows a ballistic behavior if $m(t)\propto t^2$, while it shows a diffusive behavior if $m(t)\propto t$. Moreover, $m(t)\propto t^\alpha$ with $1<\alpha<2$ indicates a supperdiffusive behavior~\cite{Shlomo1987}. \cref{Fig:smoent_exp} shows the results of {average $m(t)$ with the initial coin states $\ket{H}, \ket{V}, \ket{+}$ and $\ket{L}$ in} a ten-step QW with three different coin sets.  We observe that QWs with three coin sets exhibit supperdiffusive behavior in contrast to diffusive behavior in classical random walk [$m(t)\propto t$]. We simulate 1000 QW with the coin operation randomly selected from the coin set $\{\hat{H}, \hat{\mathds 1}\}$, and the results of $m(t)$ are shown in~\cref{Fig:smoent_exp}(b). Asymptotically, the QW with the coin set $\{\hat{H}, \hat{\mathds 1}\}$ exhibits superdiffusive behavior as $m(t)\propto t^{1.47}$, which  is slower than the ballistic behavior [$m(t)\propto t^2$] in the Hadamard walk \cite{Chandrashekar2008}. However, for simaller $t(t\leq10)$, the QW with most fixed $\bf{C}_T$ spreads faster than the Hardmard walker as shown in the insert of \cref{Fig:smoent_exp}(b). This is the reason why we observe the QW with the coin set $\{\hat{H}, \hat{\mathds 1}\}$ spreads faster than the Hardmard walk in our experiment ($T=10$) as shown in \cref{Fig:smoent_exp}(a). 

% Notably, QW with coin set $\{\hat{H}, \hat{\mathds 1}\}$ spreads faster than the other two, which is attributed to its equivalence to gEQW~\cite{Pires2019SR,Naves2022PRA}.  

In conclusion, we design coin sequences that can rigorously generate maximal entanglement between the coin and the position of the walker in a 1D DTQW with the following three key features to be available at any $T\geq3$, to be independent of initial coin state and to be the simplest for experimental implementation. A comparison of our coin sequence~$\bm{C}_T$ with the other coin sequences is shown in~\cref{App:comparison}, and MEG with our coin sequence significantly outperforms all other proposed coin sequences in the three features mentioned above. {The QW with proposed coin sequences is equivalent to gEQW, which exhibits faster spreading.}

Experimentally, we realize a ten-step 1D DTQW with proposed coin sequences, and we observe the entanglement generation as well as spreading behaviors. The results show a significant enhancement {in terms of the entanglement generation}, which benefits the intermediate quantum information processing that requires maximal qubit-qudit entanglement. 
Moreover, the spreading of probability distributions with our coin sequence {reflects a higher uniformity and faster speedup}, which is favorable and useful in various quantum algorithms and in quantum simulation of biological processes~\cite{2003_kendon_Phys.Rev.A_Decoherence,2007_maloyer_NewJ.Phys._Decoherence,2020_martin-vazquez_Phys.Rev.A_Optimizing}. Our protocol can also be generalized to a $p$-diluted disorder QW~\cite{Geraldi2019PhysRevLett,Geraldi2021PhysRevResearch}, in which transport behavior can be engineered by controlling the probability of coin operations. {As our model is equivalent to gEQW~\cite{Pires2019SR,Naves2022PRA}, a hyperballistic speedup is expected while maintaining the maximal entanglement generation.} 

{
\begin{acknowledgments}
We are grateful to two anonymous referees for providing very useful comments on an earlier version of this article. This work is supported by the Shandong Provincal Natural Science Foundation (Grants No. ZR2019MA001 and No. ZR2020JQ05), the National Natural Science Foundation of China (Grants No. 11974213 and No. 92065112), the National Key R\&D Program of China (Grant No. 2019YFA0308200), Taishan Scholar of Shandong Province (Grant No. tsqn202103013), Shenzhen Fundamental Research Program (Grant No.JCYJ20190806155211142), Shandong University Multidisciplinary Research and Innovation Team of Young Scholars (Grant No. 2020QNQT), and the Higher Education Discipline Innovation Project ('111') (Grant No.B13029). 
\end{acknowledgments}
}
\bibliography{Refer.bib}
%\input{DQWmain_revised.bbl}
% \clearpage

%\appendix
%\input{Appendix}

\appendix

\section{MEG with optimal coin sequence}\label{App:optimalsequence}
\subsection{Fourier analysis of quantum walks}
We use Fourier analysis to analyze the dynamical evolution in DTQW~\cite{Nayak2000,Brun2003, Annabestani2010PRA,Hinarejos2014}, which is defined as 
\begin{equation}
\begin{split}
\ket{k}&=\sum_{x\in\mathbb{Z}}\text{e}^{\text{i}kx}\ket{x}\\
\ket{x}&=\int_{-\pi}^{\pi}\frac{\text{d}k}{2\pi}\text{e}^{-\text{i}kx}\ket{k}.
\end{split}
\end{equation}
With a Fourier transformation, the shift operator in~\cref{eq:shiftoperator} can be expressed in momentum space as
\begin{equation}
\begin{split}
\hat{S}_\text{m}&=\left(\text{e}^{-\text{i} k}\ket0\bra{0}+\text{e}^{\text{i} k}\ket1\bra{1}\right)\otimes\ket{k}\bra{k}\\
&=\left(
\begin{array}{cc}
\text{e}^{-\text{i}k} & 0 \\
0 & \text{e}^{\text{i}k}
\end{array}
\right)\otimes\ket{k}\bra{k} .
\end{split}
\end{equation}
Accordingly, the evolution unitary operator $\hat{U}=\hat{S}(\hat{C}\otimes\hat{\mathds1}_\text{p})$ is expressed by 
\begin{equation}
\begin{split}
\hat{U}_\text{m}&=\hat{S}_\text{m}(\hat{C}\otimes\hat{\mathds1}_\text{p})\\
&=\left(
\begin{array}{cc}
\text{e}^{-\text{i}k} & 0 \\
0 & \text{e}^{\text{i}k}
\end{array}
\right)\hat{C}\otimes\ket{k}\bra{k}\sum_x\ket{x}\bra{x}\\
&=\frac1{4\pi^2}\int \text{d}k\int\text{d}k^\prime\left(
\begin{array}{cc}
\text{e}^{-\text{i}k} & 0 \\
0 & \text{e}^{\text{i}k}
\end{array}
\right)\hat{C}\otimes \sum_x\text{e}^{-\text{i}(k-k^\prime)x}\ket{k}\bra{k^\prime}\\
&=\frac1{2\pi}\int \text{d}k\left(
\begin{array}{cc}
\text{e}^{-\text{i}k} & 0 \\
0 & \text{e}^{\text{i}k}
\end{array}
\right)\hat{C}\otimes\ket{k}\bra{k}\\
&=\int\hat{C}_\text{m}\otimes\frac{\text{d}k}{2\pi}\ket{k}\bra{k},
\end{split}
\end{equation}
where we have used the orthonormalization relation
\begin{equation}
\sum_{x\in\mathbb{Z}}\text{e}^{-\text{i}(k-k^\prime)x}=2\pi\delta(k-k^\prime),
\end{equation}
and we denote $\left(
\begin{array}{cc}
\text{e}^{-\text{i}k} & 0 \\
0 & \text{e}^{\text{i}k}
\end{array}
\right)\hat{C}$ as $\hat{C}_\text{m}$.
Specifically, for the coin operations $\hat{C}=\hat{H}$ and $\hat{C}=\hat{\mathds1}$, we obtain
\begin{equation}
\hat{H}_\text{m}=\frac1{\sqrt{2}}\left(
\begin{array}{cc}
\text{e}^{-\text{i} k} & \text{e}^{-\text{i} k} \\
\text{e}^{\text{i} k} & -\text{e}^{\text{i} k}
\end{array}\right), 
\hat{\mathds1}_\text{m}=\frac1{\sqrt{2}}\left(
\begin{array}{cc}
\text{e}^{-\text{i} k} & 0 \\
0 &  \text{e}^{\text{i} k}
\end{array}\right).
\end{equation}
Thus, the dynamic evolution of the initial state $\ket{\theta,\phi}_\text{c}^\text{in}\otimes\ket0_\text{p}$ is
\begin{equation}
\begin{split}
\ket{\theta,\phi}_\text{f}&=\iint\prod_{t=1}^{T}\hat{C}_{\text{m},t}\otimes\frac{\text{d}k}{2\pi}\ket{k}\bra{k}\cdot\ket{\theta,\phi}_\text{c}^\text{in}\otimes\frac{\text{d}k^\prime}{2\pi}\ket{k^\prime}\\
&=\int\prod_{t=1}^{T}\hat{C}_{\text{m},t}\ket{\theta,\phi}_\text{c}^\text{in}\otimes\frac{\text{d}k}{2\pi}\ket{k}.
\end{split}
\end{equation}
The reduced density matrix of the coin is 
% \begin{equation}
% \begin{aligned}\label{Eq:coin_kspace} 
% \rho_\text{c}^\text{f}
% &=\tr_\text{p}\left(\ket{\theta,\phi}_\text{f}\bra{\theta,\phi}\right)\\
% &=\iint\sum_{x\in\mathbb{Z}}\prod_{t=1}^{T}\hat{C}_{\text{m},t}\ket{\theta,\phi}_\text{c}^\text{in}\bra{\theta,\phi}\hat{C}_{\text{m},t}^\dagger\frac{\text{d}k\text{d}k^\prime}{4\pi^2}\bracket{x}{k}\bracket{k^\prime}{x}\\ \nonumber
% &=\int\frac{\text{d}k}{2\pi}\prod_{t=1}^{T}\hat{C}_{\text{m},t}\ket{\theta,\phi}_\text{c}^\text{in}\bra{\theta,\phi}\hat{C}_{\text{m},t}^\dagger  
% \end{aligned}
% \end{equation}

\begin{align}\label{Eq:coin_kspace} 
\rho_\text{c}^\text{f}
&=\tr_\text{p}\left(\ket{\theta,\phi}_\text{f}\bra{\theta,\phi}\right)\\
&=\iint\sum_{x\in\mathbb{Z}}\prod_{t=1}^{T}\hat{C}_{\text{m},t}\ket{\theta,\phi}_\text{c}^\text{in}\bra{\theta,\phi}\hat{C}_{\text{m},t}^\dagger\frac{\text{d}k\text{d}k^\prime}{4\pi^2}\bracket{x}{k}\bracket{k^\prime}{x}\\ \nonumber
&=\int\frac{\text{d}k}{2\pi}\prod_{t=1}^{T}\hat{C}_{\text{m},t}\ket{\theta,\phi}_\text{c}^\text{in}\bra{\theta,\phi}\hat{C}_{\text{m},t}^\dagger  
\end{align}

\subsection{Superoperator}
We describe the evolution acting on the coin state by a superoperator $\hat{\mathcal L}$
\begin{equation}
\rho_\text{c}^{t+1}=\hat{\mathcal L}\rho_\text{c}^{t}.
\end{equation}
According to~\cref{Eq:coin_kspace}, the one step evolution is
\begin{equation}
\begin{split}
\rho_\text{c}^{t+1}=\int\frac{\text{d} k}{2 \pi} \hat{C}_\text{m}\rho_\text{c}^{t}\hat{C}_\text{m}^\dagger
\end{split}
\end{equation}

An arbitrary coin state $\rho_\text{c}^{t}$ can be described with Pauli matrices $\{\mathds1, \sigma_x, \sigma_y, \sigma_z\}$ by
\begin{equation}
\begin{split}
\rho_\text{c}^{t}&=\alpha_0\mathds1+\alpha_1\sigma_x+\alpha_2\sigma_y+\alpha_3\sigma_z\\
&=\left(
\begin{array}{cc}
\alpha_0+\alpha_3 & \alpha_1-i\alpha_2\\
\alpha_1+i\alpha_2 & \alpha_0-\alpha_3\\
\end{array}\right)\\
&=\left(
\begin{array}{cc}
\frac12+\alpha_3 & \alpha_1-i\alpha_2\\
\alpha_1+i\alpha_2 & \frac12-\alpha_3\\
\end{array}\right),
\end{split}
\end{equation} 
where $\alpha_0=\frac12$ satisfying $\tr\left(\rho\right)=2\alpha_0=1$. Using the affine map approach~\cite{Brun2003}, we represent $\hat{\mathcal L}$ as a matrix acting on the $2\times2$ $\rho_\text{c}^{t}$, which can be further expressed as a four-dimensional column vector $\bm{\alpha}_\text{c}^{t}$~\cite{Annabestani2010PRA}
\begin{equation}
\bm{\alpha}_\text{c}^{t}=
\left(
\begin{array}{c}
\frac12 \\
\alpha_1\\
\alpha_2\\
\alpha_3\\
\end{array}\right),
\end{equation}
where $\alpha_i=\frac12\tr\left(\rho\sigma_i\right)$. Along this spirit, dynamic evolution of the density matrix of the coin can be expressed as 
\begin{equation}\label{Eq:superoperator_dynamics}
\bm{\alpha}_c^\text{f}=\int\frac{\mathrm{d} k}{2 \pi} \prod_{t=1}^{T}{\hat{\mathcal{L}}_{t}} \bm{\alpha}_c^\text{in},
\end{equation} 
where $\bm{\alpha}_c^\text{in}$ and $\bm{\alpha}_c^\text{f}$ correspond to $\rho_c^\text{in}$ and $\rho_c^\text{f}$, respectively.

For $\hat{H}_\text{m}$, we have 
\begin{equation}
\begin{split}
&\hat{H}_\text{m}\rho_\text{c}^{t}\hat{H}_\text{m}^\dagger\\
&=\left(
\begin{array}{cc}
\frac12+\alpha_1 & (\alpha_3+i\alpha_2)e^{-2ik} \\
(\alpha_3-i\alpha_2)e^{2ik} & \frac12-\alpha_1
\end{array}\right)\\
&=\left(
\begin{array}{c}
\frac12 \\
\alpha_3\cos2k+\alpha_2\sin2k\\
-\alpha_2\cos2k+\alpha_3\sin2k\\
\alpha_1\\
\end{array}\right).
\end{split}
\end{equation}
Then we can calculate $\hat{\mathcal{L}}^H$ by
\begin{equation}
\hat{\mathcal{L}}^H\left(
\begin{array}{c}
\frac12 \\
\alpha_1\\
\alpha_2\\
\alpha_3\\
\end{array}\right)=
\left(
\begin{array}{c}
\frac12 \\
\alpha_2\sin2k+\alpha_3\cos2k\\
-\alpha_2\cos2k+\alpha_3\sin2k\\
\alpha_1\\
\end{array}\right),
\end{equation} 
and obtain the expression of $\hat{\mathcal{L}}^H$ as
\begin{equation}
\hat{\mathcal{L}}^H =\left(\begin{array}{cccc}
1 & 0 & 0 & 0 \\
0 & 0 & \sin 2 k & \cos 2 k \\
0 & 0 & -\cos 2 k & \sin 2 k \\
0 & 1 & 0 & 0
\end{array}\right).
\end{equation}
Similarly, the expression of $\hat{\mathcal{L}}^\mathds1$ is 
\begin{equation}
\hat{\mathcal{L}}^\mathds1 =\left(\begin{array}{cccc}
1 & 0 & 0 & 0 \\
0 & \cos 2 k  & -\sin 2 k & 0\\
0 & \sin 2 k & \cos 2 k & 0 \\
0 & 0 & 0 & 1
\end{array}\right).
\end{equation}

\subsection{Optimal coin sequence}
The optimal coin sequence $\bm{C}_T$ maps any initial coin state $\rho_\text{c}^{\text{in}}=\ket{\theta,\phi}_\text{c}^{\text{in}}\bra{\theta,\phi}$ to identity state $\rho_\text{c}^{\text{f}}=1/2$. In the context of superoperator, we have the following definition of optimal coin sequence.  
\begin{definition}[Optimal coin sequence]\label{Df:optimalsequence}
A coin sequence $\bm{C}_T$ is the optimal sequence if its corresponding superoperators can transform $\bm{\alpha}_c^\text{in}$ to $\bm{\alpha}_c^\text{f}$, where
\begin{equation} 
\bm{\alpha}_c^\text{in}=\frac12\left(\begin{array}{c}
1 \\
\cos \phi \sin \theta \\
\sin \phi \sin \theta \\
\cos \theta
\end{array}\right),
\bm{\alpha}_c^\text{f}=\frac12\left(\begin{array}{c}
1 \\
0 \\
0 \\
0
\end{array}\right).
\end{equation}
\end{definition}

We first propose the following optimal coin sequence with one Hadamard operation: 
\begin{theorem}\label{Th:1H}
Given $l_{1,2}\in\mathbb{N}$, the coin sequence with $\bm b=1^{\otimes l_1}01^{\otimes l_2}$
is optimal if $l_1$ and $l_2$ satisfy 
$l_1\ne0$ and 
$l_1 \ne l_2+1$.
\end{theorem}
\begin{proof}
Given an $l\in\mathbb{N}^+$, $\hat{\mathcal{L}}^\mathds1$ has the property of 
\begin{equation}
(\hat{\mathcal{L}}^\mathds1)^{\otimes l}=\left(
\begin{array}{cccc}
1 & 0 & 0 & 0 \\
0 & \cos (2lk) & -\sin (2lk) & 0 \\
0 & \sin (2lk) & \cos (2lk) & 0 \ \\
0 & 0 & 0 & 1
\end{array}\right). 
\end{equation}
Then we calculate $\bm{\alpha}_c^\text{f}$ by
%\begin{widetext}
 \begin{equation}
 \begin{split}
\bm{\alpha}_c^\text{f}&=\int_{-\pi}^{\pi}\frac{\mathrm{d} k}{2\pi}(\hat{\mathcal{L}}^\mathds1)^{\otimes l_2} \hat{\mathcal{L}}^H(\hat{\mathcal{L}}^\mathds1)^{\otimes l_1}\bm{\alpha}_c^\text{in}\\
&=\int_{-\pi}^{\pi}\frac{\mathrm{d} k}{2\pi}\frac12\left(
\begin{array}{c}
1\\
\alpha_1   \\
\alpha_2  \\
\alpha_3
\end{array}\right),
\end{split}
\end{equation}
%\end{widetext}
where 
\begin{equation}
\begin{split}
\alpha_1=&\cos\theta \cos\left[2 (l_2+1)k\right]\\
&+\sin\phi\sin\theta  \sin\left[2 (l_2+1)k\right]   \cos(2 l_1k)\\
&+\cos\phi\sin\theta  \sin\left[2 (l_2+1)k\right]   \sin(2 l_1k)\\
\alpha_2=&\cos\theta \sin\left[2 (l_2+1)k\right]\\
&-\sin\phi\sin\theta \cos\left[2 (l_2+1)k\right]  \cos(2 l_1 k)\\
&-\cos\phi\sin\theta  \cos\left[2 (l_2+1)k\right]  \sin(2 l_1 k)\\
 \alpha_3=&\cos\phi\sin\theta   \cos(2 l_1k) -\sin\phi\sin\theta  \sin(2 l_1 k).\\
\end{split}
\end{equation}
The momentum integrals for $\alpha_3$ is under the condition $l_1\ne 0$, and that of $\alpha_1$ and $\alpha_2$ are 0 as well under the condition $l_1\neq l_2+1$. Then the coin sequence with $\bm b=1^{\otimes l_1}01^{\otimes l_2}$ transforms $\bm{\alpha}_c^\text{in}$ to $\bm{\alpha}_c^\text{f}$ satisfying~\cref{Df:optimalsequence}, and is the optimal coin sequence. 
\end{proof}

We then propose another optimal coin sequence with two Hadamard operations: 
\begin{theorem}\label{Th:2H}
Given $l_{1, 2, 3}\in\mathbb{N}$, the coin sequence with $\bm b=1^{\otimes l_1}01^{\otimes l_2}0 1^{\otimes l_3}$ is the optimal coin sequence if $l_1$, $l_2$ and $l_3$ satisfy $l_1 \ne l_2+1$, $l_1 \ne l_3+1$, $l_2 \pm l_1\ne l_3$ and $l_2+l_3-l_1\neq2$.
\end{theorem}
\begin{proof}
$\bm{\alpha}_c^\text{f}$ is calculated by
 \begin{equation}
 \begin{split}
\bm{\alpha}_c^\text{f}&=\int_{-\pi}^{\pi}\frac{\mathrm{d} k}{2\pi}(\hat{\mathcal{L}}^\mathds1)^{\otimes l_3} \hat{\mathcal{L}}^H(\hat{\mathcal{L}}^\mathds1)^{\otimes l_2} \hat{\mathcal{L}}^H(\hat{\mathcal{L}}^\mathds1)^{\otimes l_1}\bm{\alpha}_c^\text{in}\\
&=\int_{-\pi}^{\pi}\frac{\mathrm{d} k}{2\pi}\frac12\left(
\begin{array}{c}
1\\
\alpha_1   \\
\alpha_2  \\
\alpha_3
\end{array}\right)
\end{split}
\end{equation}
with
\begin{widetext}
\begin{equation}
\begin{split}
\alpha_1=&\cos\theta \sin\left[2 (l_2+1)k\right]\sin\left[2 (l_3+1)k\right]\\
&-\sin\phi\sin\theta  \bigg(\cos(2l_1k)\cos\left[2 (l_2+1)k\right] \sin\left[2 (l_3+1)k\right]+\sin(2l_1k)\cos\left[2 (l_3+1)k\right]\bigg)\\
&-\cos\phi\sin\theta  \bigg(\sin(2 l_1k)\cos\left[2 (l_2+1)k\right]\sin\left[2 (l_3+1)k\right]-\cos(2l_1k)\cos\left[2 (l_3+1)k\right]\bigg)\\
\alpha_2=&-\cos\theta \sin\left[2 (l_2+1)k\right]\cos\left[2 (l_3+1)k\right]\\
&-\sin\phi\sin\theta  \bigg(\cos(2l_1k)\cos\left[2 (l_2+1)k\right] \cos\left[2 (l_3+1)k\right]+\sin(2l_1k)\sin\left[2 (l_3+1)k\right]\bigg)\\
&-\cos\phi\sin\theta  \bigg(\sin(2 l_1k)\cos\left[2 (l_2+1)k\right]\cos\left[2 (l_3+1)k\right]-\cos(2l_1k)\cos\left[2 (l_3+1)k\right]\bigg)\\
 \alpha_3=&\cos\theta\cos\left[2 (l_2+1)k\right]+\sin\phi\sin\theta   \cos(2 l_1k)\sin\left[2 (l_2+1)k\right] +\cos\phi\sin\theta  \sin(2 l_1 k)\sin\left[2 (l_2+1)k\right].\\
\end{split}
\end{equation}
\end{widetext}
The momentum integrals for $\alpha_3$ is under the condition $l_1\neq l_2+1$, and that of $\alpha_1$ and $\alpha_2$ are 0 under the conditions $l_2\pm l_1\neq l_3, l_2+l_3-l_1\neq2$ and $l_1\neq l_3+1$.
\end{proof}

The coin operation $\hat{C}_1$ at the first step ($T=1$) is equivalent to changing the initial coin state so that we have the following Corollary:   
\begin{corollary}\label{Co:1}
Given $l_1, l_2, l_3 \in\mathbb{N}$ satisfying~\cref{Th:1H} and~\cref{Th:2H}, the coin sequences with  $\bm b=01^{\otimes (l_1-1)}01^{\otimes l_2}$ and $\bm b=01^{\otimes (l_1-1)}01^{\otimes l_2}0 1^{\otimes l_3}$ are the optimal sequences. 
\end{corollary}

{Note that optimal $\bm{b}$ in~\cref{Th:1H}, \cref{Th:2H} and \cref{Co:1} is not limited to generating the maximal entanglement step $T$, but also leads to MEG in the intermediate steps $t$ if the corresponding sequence $\bm b=b_1b_2\cdots b_t$ fulfills the conditions in~\cref{Th:1H}, \cref{Th:2H} or \cref{Co:1}. We experimentally set the coin sequences $\bm{C}_T$ according to~\cref{Co:1}; they are listed in~\cref{Tb:coinsequences}. }
\begin{table}[h!]
\begin{center}
\begin{tabular}{ |c|c|c| } 
 \hline
~$T$ &  $\bm{C}_T$ with coin set: $\{\hat{H},\hat{\mathds1}\}$ &  $\bm{C}_T$ with coin set: $\{\hat{H},\hat{F}\}$ \\\hline 
3 & $\{\hat{H},\hat{H},\hat{\mathds1}\}$ & $\{\hat{F},\hat{F},\hat{H}\}$\\ 
4  & $\{\hat{H},\hat{H},\hat{\mathds1}, \hat{\mathds1}\}$& $\{\hat{H},\hat{F},\hat{H}, \hat{H}\}$ \\ 
5 & $\{\hat{H},\hat{H},\hat{\mathds1}, \hat{\mathds1}, \hat{\mathds1}\}$ & $\{\hat{H},\hat{F},\hat{H}, \hat{F}, \hat{F}\}$  \\ 
6 & $\{\hat{H},\hat{H},\hat{\mathds1}, \hat{\mathds1}, \hat{\mathds1},\hat{\mathds1}\}$ & $\{\hat{F},\hat{H},\hat{F}, \hat{H}, \hat{H},\hat{F}\}$ \\ 
7 & $\{\hat{H},\hat{H},\hat{\mathds1}, \hat{H},\hat{\mathds1}, \hat{\mathds1},\hat{\mathds1}\}$ & $\{\hat{H},\hat{H},\hat{F}, \hat{H},\hat{F}, \hat{F},\hat{H}\}$  \\ 
8 & $\{\hat{H},\hat{H},\hat{\mathds1}, \hat{H},\hat{\mathds1}, \hat{\mathds1},\hat{\mathds1},\hat{\mathds1}\}$ & $\{\hat{H},\hat{F},\hat{F}, \hat{F},\hat{F}, \hat{H},\hat{F},\hat{F}\}$ \\ 
9 & $\{\hat{H},\hat{H},\hat{\mathds1}, \hat{H},\hat{\mathds1}, \hat{\mathds1},\hat{\mathds1},\hat{\mathds1},\hat{\mathds1}\}$ & $\{\hat{F},\hat{H},\hat{F}, \hat{F},\hat{H}, \hat{H},\hat{H},\hat{H},\hat{F}\}$  \\ 
10 & $\{\hat{H},\hat{H},\hat{\mathds1}, \hat{H},\hat{\mathds1}, \hat{\mathds1},\hat{\mathds1},\hat{\mathds1},\hat{\mathds1},\hat{\mathds1}\}$ & $\{\hat{F},\hat{F},\hat{H}, \hat{H},\hat{H}, \hat{F},\hat{F},\hat{F},\hat{H},\hat{H}\}$ \\ 
 \hline
\end{tabular}
\end{center}
\caption{Experimental settings of coin sequences $\bm{C}_T$ with coin sets $\{\hat{H},\hat{\mathds1}\}$ and $\{\hat{H},\hat{F}\}$.}\label{Tb:coinsequences}
\end{table}

{
\section{Generalized elephant quantum walk}\label{App:gEQW}
In this section, we briefly introduce the gEQW in Ref.~\cite{Pires2019SR,Naves2022PRA}. The main difference between the Hadamard QW and the gEQW is that the shift operator in~\cref{Eq:shiftoperator_gEQW}. The probability of $\Delta_t$ is determined by the discrete version of the $q$-exponential distribution in $[1,T]$ as  
\begin{equation}\label{Eq:qfunction}
\text{Pr}(\Delta_t)=e_q^{\Delta_t}=\tau_t[1-(1-q)\Delta_t]^{1/1-q}
\end{equation}
with $\tau_t$ being a time-dependent normalization factor. The support of function~\cref{Eq:qfunction} is given by
\begin{equation}\label{Eq:support}
\text{supp}[e_q(x)]=
\begin{cases}
[0,\frac{1}{1-q}]&q\leq1\\
[0,\infty]&q>1
\end{cases}.
\end{equation}
When $q=\frac{1}{2}$, we obtain $\Delta_t\leq2$ with probability distribution of $\text{Pr}(\Delta_t=1)=1$ and $\text{Pr}(\Delta_t=2)=0$, which is the Hadamard QW; when $q\to1$, we obtain a decreasing exponential $\text{Pr}(\Delta_t)=\tau_te^{-\Delta_t}$; and when $q\to\infty$, we obtain an uniform distribution $\text{Pr}(\Delta_t)=\frac{1}{T}$. 

The QWs with the optimal coin sequence in~\cref{Co:1} are equivalent to the gEQW, i.e., $\bm b=01^{\otimes (l_1-1)}01^{\otimes l_2}$ corresponds to a two-step gEQW with $\Delta_1=l_1$ and $\Delta_2=l_2+1$, and $\bm b=01^{\otimes (l_1-1)}01^{\otimes l_2}0 1^{\otimes l_3}$ corresponds to a three-step gEQW with $\Delta_1=l_1, \Delta_2=l_2+1$ and $\Delta_3=l_3+1$. In this sense, QWs with the optimal coin sequence in~\cref{Co:1} are gEQW with specific shift-operator configurations, for which the maximal entanglement generation in the asymptotic approach has been reported~\cite{Naves2022PRA}. }

\section{More experimental results}\label{App:processodd}
{The reconstructed $\chi^{\text{exp}}_{\bm{C}_T}$ with the coin sequences in~\cref{Tb:coinsequences} are shown in~\cref{Eq:chi_HI}, according to which we calculate the process fidelities $\mathcal F_{\bm C_3}=0.9976\pm0.0003, \mathcal F_{\bm C_4}=0.9955\pm0.0016, \mathcal F_{\bm C_5}=0.9939\pm0.0011, \mathcal F_{\bm C_6}=0.9951\pm0.0009, \mathcal F_{\bm C_7}=0.9975\pm0.0004, \mathcal F_{\bm C_8}=0.9885\pm0.0015, \mathcal F_{\bm C_9}=0.9968\pm0.0004$ and $\mathcal F_{\bm C_{10}}=0.9980\pm0.0003$.} The geometric representations of the reconstructed $\chi^{\text{exp}}_{\bm{C}_T}$ with coin sets $\{\hat{H}, \hat{\mathds1}\}$ and $\{\hat{H}, \hat{F}\}$ at $T=3,5,7$ and $9$ are shown in~\cref{Fig:Geometricplot2}.\\
The reconstructed $\chi^{\text{exp}}_{\bm{C}_T}$ are as follows:
\begin{figure*}[ht]%4.25
 	\includegraphics[width=0.6\linewidth]{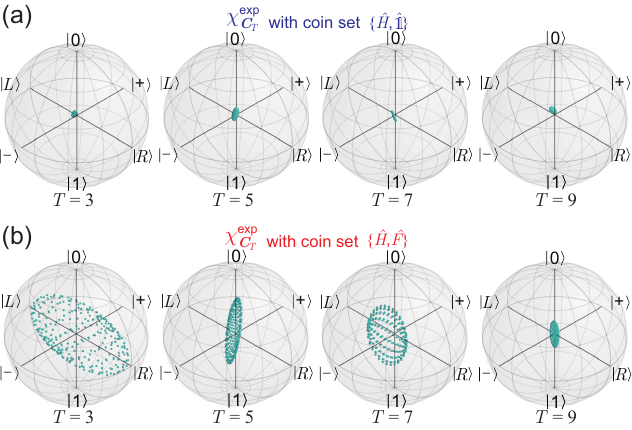} 
 	\caption{(a) Geometric representation of the reconstructed $\chi^{\text{exp}}_{\bm{C}_T}$ with the coin set $\{\hat{H}, \hat{\mathds1}\}$ at $T=3,5,7$ and $9$. (b) Geometric representation of the reconstructed $\chi^{\text{exp}}_{\bm{C}_T}$ with the coin set $\{\hat{H}, \hat{F}\}$  at $T=3,5,7$ and $9$.} 
 	\label{Fig:Geometricplot2}
 \end{figure*}

\onecolumngrid
\begin{equation}
\label{Eq:chi_HI}
\begin{aligned}
\chi^{\text{exp}}_{\bm{C}_3}&=
\left(\begin{array}{cccc}
0.2580-0.0000i &0.0013-0.0003i &-0.0135-0.0118i &0.0050-0.0006i\\
0.0013+0.0003i &0.2320-0.0000i &-0.0084+0.0020i &-0.0163+0.0042i\\
-0.0135+0.0118i &-0.0084-0.0020i &0.2520-0.0000i &0.0074+0.0046i\\
0.0050+0.0006i &-0.0163-0.0042i &0.0074-0.0046i &0.2579-0.0000i\\
\end{array}\right),\\
\chi^{\text{exp}}_{\bm{C}_4}&=
\left(\begin{array}{cccc}
0.2461-0.0000i &-0.0046+0.0056i &0.0101+0.0059i &0.0028-0.0145i\\
-0.0046-0.0056i &0.2518-0.0000i &-0.0082-0.0092i &-0.0327+0.0011i\\
0.0101-0.0059i &-0.0082+0.0092i &0.2408-0.0000i &0.0062-0.0037i\\
0.0028+0.0145i &-0.0327-0.0011i &0.0062+0.0037i &0.2613-0.0000i\\
\end{array}\right),\\
\chi^{\text{exp}}_{\bm{C}_5}&=
\left(\begin{array}{cccc}
0.2441-0.0000i &0.0165+0.0087i &0.0071-0.0204i &0.0133+0.0021i\\
0.0165-0.0087i &0.2524-0.0000i &-0.0093+0.0090i &-0.0366-0.0049i\\
0.0071+0.0204i &-0.0093-0.0090i &0.2405-0.0000i &-0.0001+0.0054i\\
0.0133-0.0021i &-0.0366+0.0049i &-0.0001-0.0054i &0.2630-0.0000i\\
\end{array}\right),\\
\chi^{\text{exp}}_{\bm{C}_6}&=
\left(\begin{array}{cccc}
0.2728-0.0000i &-0.0247-0.0021i &-0.0038-0.0147i &-0.0046-0.0035i\\
-0.0247+0.0021i &0.2594-0.0000i &0.0099+0.0002i &-0.0138+0.0025i\\
-0.0038+0.0147i &0.0099-0.0002i &0.2298-0.0000i &0.0061+0.0188i\\
-0.0046+0.0035i &-0.0138-0.0025i &0.0061-0.0188i &0.2380-0.0000i\\
\end{array}\right),\\
\chi^{\text{exp}}_{\bm{C}_7}&=
\left(\begin{array}{cccc}
0.2333-0.0000i &0.0051-0.0009i &0.0005-0.0085i &-0.0044-0.0086i\\
0.0051+0.0009i &0.2634-0.0000i &0.0049-0.0041i &-0.0191+0.0035i\\
0.0005+0.0085i &0.0049+0.0041i &0.2642-0.0000i &-0.0040+0.0050i\\
-0.0044+0.0086i &-0.0191-0.0035i &-0.0040-0.0050i &0.2392-0.0000i\\
\end{array}\right),\\
\chi^{\text{exp}}_{\bm{C}_8}&=
\left(\begin{array}{cccc}
0.2327-0.0000i &0.0146+0.0030i &0.0267-0.0002i &0.0214-0.0324i\\
0.0146-0.0030i &0.2353-0.0000i &0.0191+0.0321i &0.0211-0.0020i\\
0.0267+0.0002i &0.0191-0.0321i &0.2461-0.0000i &0.0038+0.0123i\\
0.0214+0.0324i &0.0211+0.0020i &0.0038-0.0123i &0.2859-0.0000i\\
\end{array}\right),\\
\chi^{\text{exp}}_{\bm{C}_9}&=
\left(\begin{array}{cccc}
0.2563-0.0000i &-0.0141+0.0002i &-0.0144+0.0135i &0.0103-0.0064i\\
-0.0141-0.0002i &0.2480-0.0000i &-0.0112+0.0124i &0.0110+0.0101i\\
-0.0144-0.0135i &-0.0112-0.0124i &0.2524-0.0000i &-0.0066+0.0063i\\
0.0103+0.0064i &0.0110-0.0101i &-0.0066-0.0063i &0.2433-0.0000i\\
\end{array}\right),\\
\chi^{\text{exp}}_{\bm{C}_{10}}&=
\left(\begin{array}{cccc}
0.2596-0.0000i &-0.0011-0.0055i &0.0026-0.0127i &-0.0037+0.0004i\\
-0.0011+0.0055i &0.2384-0.0000i &-0.0031-0.0117i &-0.0156-0.0060i\\
0.0026+0.0127i &-0.0031+0.0117i &0.2516-0.0000i &-0.0036+0.0065i\\
-0.0037-0.0004i &-0.0156+0.0060i &-0.0036-0.0065i &0.2503-0.0000i\\
\end{array}\right).\\
\end{aligned}
\end{equation}

\section{A comparison of our coin sequence with the other coin sequences.}\label{App:comparison}
In the context of MEG, a comparison of our coin sequence $\bm{C}_T$ with the other coin sequences including the disordered coin sequence~\cite{Vieira2013PhysRevLett}, Parrondo sequences~\cite{Dinesh2022}, three coin sequences proposed by Gratsea \etal~\cite{Gratsea_2020_Generation_of_hybrid, Gratsea_2020_Universal_optimal_coinSequences}, the position-inhomogeneous coin sequence~\cite{Zhang2022PRA} and the gEQW with disordered shift operations~\cite{Pires2019SR,Naves2022PRA} is shown in~\cref{Tb:comparasion}.

\begin{table*}[b!]
\renewcommand\arraystretch{1.5}
\centering
\begin{threeparttable}
    \begin{tabular}{|p{2.5cm}<{\centering}|p{1.5cm}<{\centering}|p{2.1cm}<{\centering}|p{2.1cm}<{\centering}|p{2cm}<{\centering}|p{2.2cm}<{\centering}|p{3cm}<{\centering}|}
%\begin{tabular}{|>{\centering\arraybackslash}m{0.15\linewidth}|L|>{\centering\arraybackslash}m{0.12\linewidth}|>{\centering\arraybackslash}m{0.15\linewidth}|>{\centering\arraybackslash}m{0.12\linewidth}|L|}
    \toprule[1pt]
%    \multicolumn{6}{|c|}{Entanglement generation in QWs} \\
%    \specialrule{0em}{1pt}{1pt}
%    \backslashbox[3cm]{Protocols}{Properties}
     & Steps\tnote{1} & Independence\tnote{2} & Number of operations\tnote{3}  &Technique\tnote{4}  & Figure of merit & Experimental demonstrations\\ 
     \midrule[0.5pt]
    This work & $T\geq 3$ & Yes & $N_c=1$; $N_s=1$ & Annealing and Fourier analysis  & $\mathcal F_{\bm{C}_T}$ & Linear optics (this work) \\ 
     \midrule[0.5pt]
    Vieira \etal~\cite{Vieira2013PhysRevLett} & $T\to\infty$ & Yes & $N_c=2$; $N_s=1$ & Disorder &$\braket{\mathcal S_{\text{E}}}$&Linear optics (Wang \etal~\cite{Wang2018Optica}) \\ 
    \midrule[0.5pt]
     Govind \etal~\cite{Dinesh2022} & $T=3, 5$ & Partially\tnote{a} & $N_c=2$; $N_s=1$& Parrondo sequences &$\braket{\mathcal S_{\text{E}}}$& None\\ 
    \midrule[0.5pt]
    Gratsea \etal ~\cite{Gratsea_2020_Generation_of_hybrid} & $T\geq1$  &No  & Full set of SU(2) coins;$N_s=1$ & Basin hopping &$\braket{\mathcal S_{\text{E}}}$ and Inverse participation ration& Linear optics (Tao \etal~\cite{Tao2021OL}) \\
    \midrule[0.5pt]
    Gratsea \etal ~\cite{Gratsea_2020_Universal_optimal_coinSequences} & None & Partially\tnote{b}  &$N_c=2$; $N_s=1$ & Reinforcement learning &$\braket{\mathcal S_{\text{E}}}$& None\\
     \midrule[0.5pt]
    Zhang \etal ~\cite{Zhang2022PRA} & Odd steps & No  &$N_c=2$; $N_s=1$ & Numerical &$\braket{\mathcal S_{\text{E}}}$& Linear optics (Zhang \etal ~\cite{Zhang2022PRA})\\
    \midrule[0.5pt]
    Pires \etal~\cite{Pires2019SR}; Naves \etal ~\cite{Naves2022PRA}& $T\to\infty$ & Almost  &$N_c=1$; $N_s\to\infty$ & $q$-exponential distribution &$\braket{\mathcal S_{\text{E}}}$& Our demonstration corresponds to a specific case of this protocol\\
    \bottomrule[0.5pt]
    \end{tabular}

    \begin{tablenotes}
    \footnotesize
    \item[1]The steps that the average von Neumann entropy $\braket{\mathcal S_{\text{E}}}=1$ or the process fidelity $\mathcal F_{\bm{C}_T}=1$ is fulfilled.
    \item[2]The independence of the initial coin state $\ket{\theta,\phi}_\text{c}^{\text{in}}=\cos(\nicefrac\theta2)\ket0_\text{c}+\text{e}^{\text{i}\phi}\sin(\nicefrac\theta2)\ket1_\text{c}$.
    \item[3]The number of coin operations $N_c$ in the coin sequence $\bm{C}_T$ and the number of shift operations $N_s$.
    \item[4]Techniques and algorithms to determine the coin set or the coin sequence $\bm{C}_T$.
    \item[a]Independent of $\phi$
    \item[b]Universal sequence: independent of $\theta$ when $\phi=0$; Optimal sequence: dependent on $\theta$ and $\phi$.
    \end{tablenotes}
    \end{threeparttable}
\caption{%
   Comparison of performances of entanglement generation with different coin sequences.}\label{Tb:comparasion}
\end{table*}

\end{document}